\documentclass[3p,authoryear]{elsarticle}
\usepackage{graphicx} % Required for inserting images
\usepackage{amssymb,amsmath,amsthm}
\usepackage[hidelinks]{hyperref}
\newtheorem{definition}{Definition}
\newtheorem{theorem}{Theorem}
\newtheorem{lemma}{Lemma}
\newtheorem{proposition}{Proposition}
\newtheorem{claim}{Claim}
\usepackage{wrapfig}
\usepackage{lipsum} % for dummy text
\renewcommand{\phi}{\varphi}
\newenvironment{proof-of-claim}{\begin{trivlist}\item\noindent{\em Proof of Claim.}}{\hfill {\footnotesize \qed}\end{trivlist}}
\usepackage{subcaption}

\begin{document}
\begin{frontmatter}

%\title{Peer Oversight in Decision Mechanisms}
\title{Peer Oversight in Collective Decision Making}
%\title{Peer Oversight: Trustworthiness by Design}
\author[soton]{Sarah Mohsen}
\ead{szrm1g23@soton.ac.uk}

\author[soton]{Pavel Naumov}
\ead{p.naumov@soton.ac.uk}

\address[soton]{School of Electronics and Computer Science\\ University of Southampton, United Kingdom}
\begin{abstract}
This article introduces peer $k$-oversight, a property of sequential collective decision mechanisms requiring at least $k$ agents to be responsible for every harmful outcome. It is shown that whenever $k$-oversight can be achieved by redistributing control over the decisions in a mechanism, it can be achieved using just $k$ agents. A polynomial-time algorithm is also presented that determines whether such a redistribution exists and, when it does, constructs one. These results establish peer oversight as a tractable design principle for multiagent decision-making systems.
\end{abstract}

\begin{keyword}
counterfactual responsibility \sep responsibility gap \sep mechanism design \sep multiagent systems \sep collective decision making
\end{keyword}

\end{frontmatter}

\section{Introduction}

In this article, we propose and study {\em peer oversight}, a formally defined concept of shared responsibility for a harmful outcome of a collective decision. We start with a discussion of the notion of responsibility. 

\subsection{Responsibility and Strategic Abilities}\label{20th century}

Consider a variation of Joe Halpern's fish and polluting factories example~\citeyearpar[p. 40, Example 2.3.8]{h16}. In our version of this example, set in the last century, there are three lakes (A, B, and C) and three factories (F1, F2, F3) located at the shores of these lakes, see Figure~\ref{intro-lakes-zero subfigure}. Each factory has accumulated a pollutant and must dump it into one of the two lakes it is connected to in the figure. For instance, factory F1 can dump either into Lake A or Lake B. If any two factories dump the pollutant into the same lake, the fish in that lake die. If all three of them dump the pollutant into different lakes, the fish survive.

By a {\em decision mechanism}, we mean the description of the process used to make a collective decision. First, let us consider a {\em concurrent} decision mechanism under which each factory chooses its action (the lake where it will dump the pollutant) independently and without knowing the actions chosen by the other factories. This mechanism is used in strategic (or normal form) games.

Suppose that factories F1 and F2 dump the pollutant into Lake B and the fish in Lake B die. Who is responsible for this? Responsibility is a vague concept that can be formally defined in many different ways. The most popular~\citep{w17} of them is based on Frankfurt's~\citeyearpar{f69tjop} principle of alternative possibilities\footnote{It is worth noting that~\citep{f69tjop} discusses limitations of this definition and constructs examples when the definition fails to capture the intuitive notion of moral responsibility.}:
{\em \dots\ a person is morally responsible for what he has done only if he could have done otherwise}. In the literature, ``could have done otherwise'' has been interpreted as having a strategy to avoid the harmful outcome no matter what the other agents do~\citep*{bd13clima,ydjal19aamas,nt20aaai,bfm21ijcai,s24aaai,sn25jpl,sn26aaai}.  

In this article, we refer to the responsibility defined through this principle as {counterfactual responsibility} or just {\em responsibility}. \cite{ch04jair} defined a degree of responsibility in causal models using a similar counterfactual approach.  

When applying the definition of counterfactual responsibility to our setting, it is important to specify what we mean by the ``harmful outcome''. If the harmful outcome consists of the death of the fish specifically in Lake B, then factory F1 had an individual strategy to avoid such an outcome by dumping the pollutant into Lake A. Similarly, factory F2 could avoid the death of the fish in Lake B by dumping the pollutant into Lake C. Both of these strategies avoid the death of the fish in Lake B by endangering the life of the fish at another lake. To avoid such trolley-like dilemmas, in this article, by ``harm'' we mean the death of the fish in any of the three lakes. In other words, a strategy that avoids the harmful outcome must guarantee the survival of the fish in all three lakes. 

\begin{figure}%[ht]
    \centering
    \begin{subfigure}[b]{0.39\textwidth}
        \begin{center}
        \scalebox{0.5}{\includegraphics{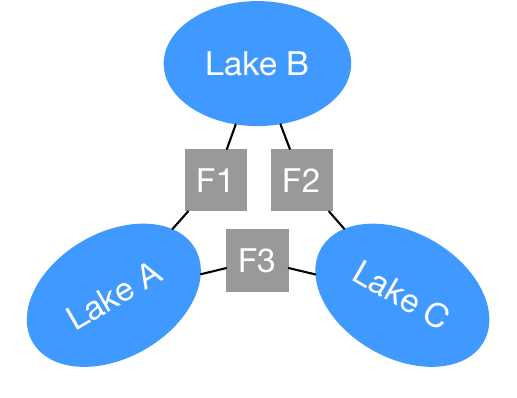}}
        \caption{Locations of factories}
       \label{intro-lakes-zero subfigure}
        \end{center}
    \end{subfigure}
    \hfill
    \begin{subfigure}[b]{0.59\textwidth}
        \begin{center}
       \scalebox{0.5}{\includegraphics{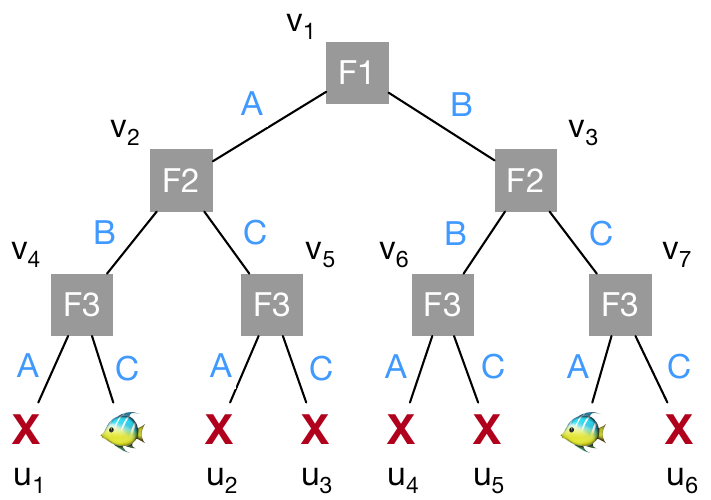}}
        \caption{Decision mechanism}
        \label{intro-lakes-tree subfigure}
        \end{center}
    \end{subfigure}
\caption{Three factories example.}
\end{figure}
It is easy to see that none of the three factories has a strategy that would guarantee the survival of the fish in all three lakes. Indeed, note, for example, that no matter what action factory F1 chooses, factories F2 and F3 can choose to dump their portions of the pollutant into Lake C and, thus, kill the fish in that lake. Hence, in our example, none of the three factories is responsible for the harmful outcome. Moreover, in this setting, no matter what action either of the factories chooses, none of them is ever responsible for the harmful outcome. 
If one is aiming to enforce individual responsibility in the fish and three factories setting, then a concurrent decision mechanism is probably the worst option to use!

Let us now consider, for the same setting, a {\em sequential} decision mechanism under which the factories choose their actions in the order F1, F2, and F3. We visualise this mechanism as a tree depicted in Figure~\ref{intro-lakes-tree subfigure}. The labels on the edges denote the lakes that the acting agents decided to dump the pollutant in. The leaf nodes of this tree represent outcomes. We label the harmful outcomes $u_1$ through $u_6$ with the red letter X and the rest with the fish emoji. 

Consider the {\em decision path} $v_1,v_2,v_4,u_1$ in Figure~\ref{intro-lakes-tree subfigure}. It corresponds to the scenario under which, first, factory F1 dumps the pollutant into Lake A; then, factory F2 dumps into Lake B; finally, F3 dumps into Lake A. Because two factories dumped into Lake A, the fish in this lake die. Note that factory F3 did not have a strategy to save the fish in all three lakes {\em upfront}, but it has such a strategy (dump into Lake C) at node $v_4$ along the decision path. Thus, factory F3 is (counterfactually) responsible for the death of the fish along the path $v_1,v_2,v_4,u_1$. It is easy to see that at any node along this path neither factory F1 nor factory F2 has a strategy to save the fish. Hence, neither of them is responsible along this path. Similarly, factory F3 is the only factory responsible along the decision path from node $v_1$ to node $u_6$.

Next, let us consider the decision path from node $v_1$ to node $u_2$. Neither of the three agents has an individual strategy to avoid harm at any of the nodes along this path. Thus, neither of them is (counterfactually) responsible along the path for the death of the fish. We say that this path belongs to the {\em responsibility gap} of the decision mechanism in Figure~\ref{intro-lakes-tree subfigure}. In the literature, the responsibility gap is also sometimes referred to as a responsibility void. One might argue that factory F2 should be blamed along the path from node $v_1$ to node $u_2$ because, at $v_2$, it has chosen the action (dump into Lake C) that made the death of the fish in at least one of the lakes unavoidable. This argument captures a different form of responsibility that we discuss in Section~\ref{related work}.

\begin{figure}%[ht]
    \centering
    \begin{subfigure}[b]{0.45\textwidth}
        \begin{center}
        \scalebox{0.5}{\includegraphics{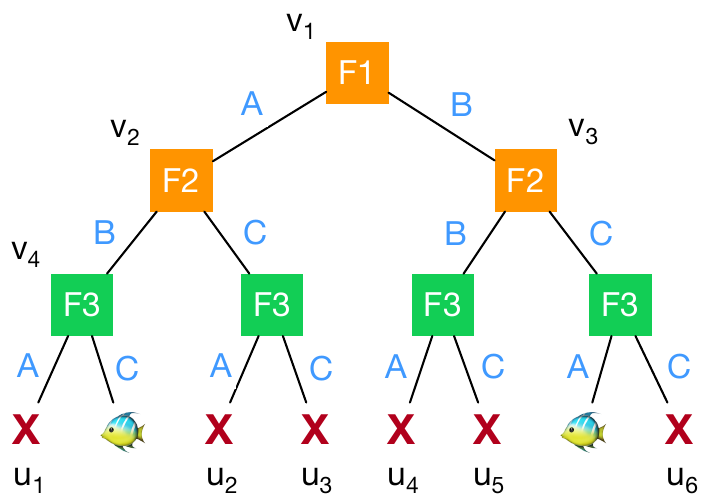}}
        \caption{Company Amber has an upfront strategy that guarantees the death of the fish in Lake B.}
       \label{intro-lakes-tree-colours subfigure}
        \end{center}
    \end{subfigure}
    \hfill
    \begin{subfigure}[b]{0.45\textwidth}
        \begin{center}
       \scalebox{0.5}{\includegraphics{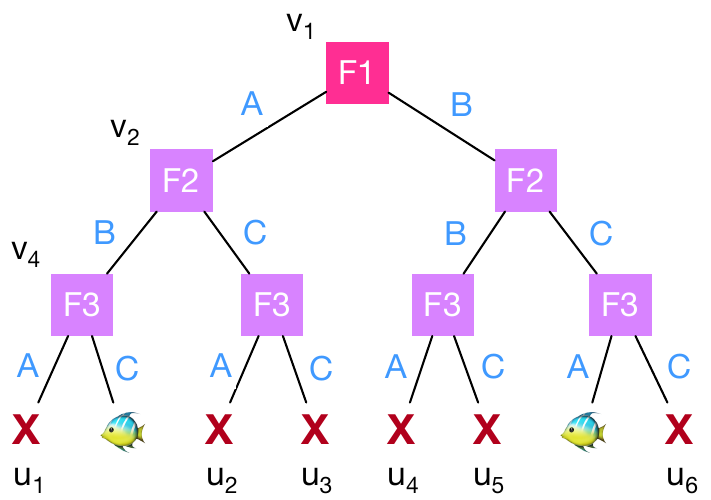}}
        \caption{Company Purple has a strategy to save the fish
        along each decision path leading to a harmful outcome.}
        \label{intro-lakes-tree-colours-two subfigure}
        \end{center}
    \end{subfigure}
\caption{Three factories example.}
\end{figure}
The counterfactual responsibility gap can often be eliminated by combining multiple parties into a single agency. To illustrate this, let us first suppose factories F1 and F2 are controlled by the management of company Amber and factory F3 is controlled by company Green, see Figure~\ref{intro-lakes-tree-colours subfigure}. By combining the control over factories F1 and F2 in the hands of company Amber, we endow the company with strategic abilities that neither of the two factories has alone. For example, company Amber has an {\em upfront} strategy to guarantee the death of the fish in Lake B. The strategy is to dump the pollutant into Lake B at nodes $v_1$ and $v_3$, see Figure~\ref{intro-lakes-tree-colours subfigure}. However, under the distribution of the control over the factories between companies Amber and Green that we consider, neither of the companies has a strategy to avoid the harm along the decision paths to outcomes $u_2$, $u_3$, $u_4$, and $u_5$. Thus, combining the control over the factories F1 and F2 in the hands of company Amber does not eliminate the responsibility gap. 

The situation is different in Figure~\ref{intro-lakes-tree-colours-two subfigure}, where company Pink controls factory F1 and company Purple controls factories F2 and F3. In this setting, company Purple has an upfront strategy to avoid harm no matter what company Pink does. Indeed, if company Pink instructs factory F1 to dump the pollutant into Lake A, then company Purple should instruct factories F2 and F3 to dump into lakes B and C, respectively. If company Pink instructs factory F1 to dump the pollutant into Lake B, then company Purple should instruct factories F2 and F3 to dump into lakes C and A, respectively. In either case, the fish in all three lakes survive. This upfront strategy is the strategy at the root node $v_1$ of the tree. Thus, company Purple has a strategy to avoid harm at some node along each decision path leading to a harmful outcome. Hence, company Purple is responsible along each such path. Therefore, the responsibility gap of the decision mechanism depicted in Figure~\ref{intro-lakes-tree-colours-two subfigure} is empty. We say that this mechanism is {\em gap-free}. 

Together, Figures~\ref{intro-lakes-tree subfigure}, \ref{intro-lakes-tree-colours subfigure}, and \ref{intro-lakes-tree-colours-two subfigure} show that a well-designed aggregation of the control over individual decisions in the hands of larger agencies can potentially eliminate the responsibility gap of the decision mechanism.

\subsection{Peer Oversight}

From the last century example of Section~\ref{20th century}, let us now transition to the present day, where the industrial area grew to a six-factory complex, depicted in Figure~\ref{intro-lake-one subfigure}. All six factories are owned by the same multinational company and discharge the pollutant in the sequential order F1 through F6. The company employs a single AI agent (named Red) to decide which lake each factory dumps the pollutant into. Note that each lake is connected to four factories. We assume that, due to the improved environmental standards, {\em all four} factories must dump the pollutant into a lake in order to kill the fish in that lake. 

It is easy to see that agent Red has an upfront strategy to avoid harm by spreading the pollutant evenly between the three lakes. Thus, agent Red is responsible for the harm each time the fish die. Hence, this mechanism is gap-free.

AI agents have a tendency to be unreliable. To improve {\em trustworthiness} of the decision-making process, the company decides to replace the single-agent decision-making system with a multi-agent AI system consisting of agents Red, Amber, and Green. The Red agent chooses the lake for factories F1 and F4, the Amber agent chooses for F2 and F5, and the Green agent does this for F3 and F6, see Figure~\ref{intro-lakes-two subfigure}. The company made this change because of the following observation.

\begin{figure}%[ht]
    \centering
    \begin{subfigure}[b]{0.32\textwidth}
        \begin{center}
        \scalebox{0.5}{\includegraphics{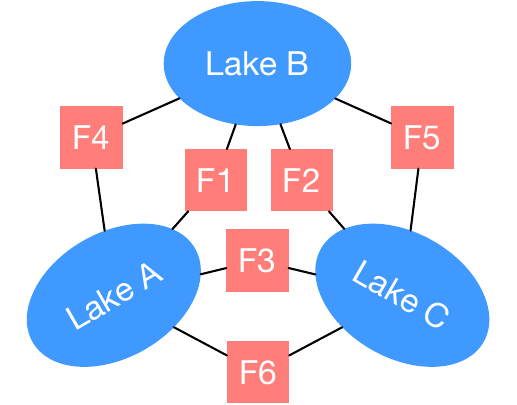}}
        \caption{Gap-free mechanism with a single agent}
       \label{intro-lake-one subfigure}
        \end{center}
    \end{subfigure}
    \hfill
    \begin{subfigure}[b]{0.32\textwidth}
        \begin{center}
       \scalebox{0.5}{\includegraphics{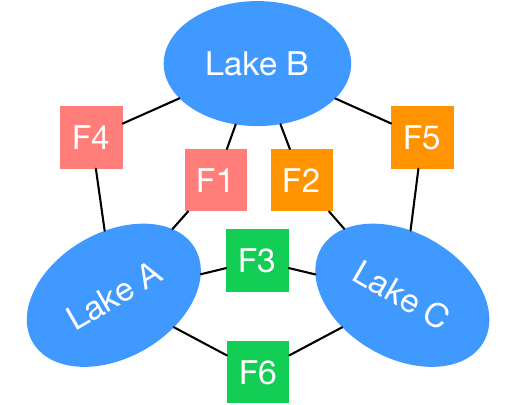}}
        \caption{Two-oversight by three agents}
        \label{intro-lakes-two subfigure}
        \end{center}
    \end{subfigure}
    \hfill
    \begin{subfigure}[b]{0.32\textwidth}
        \begin{center}
       \scalebox{0.5}{\includegraphics{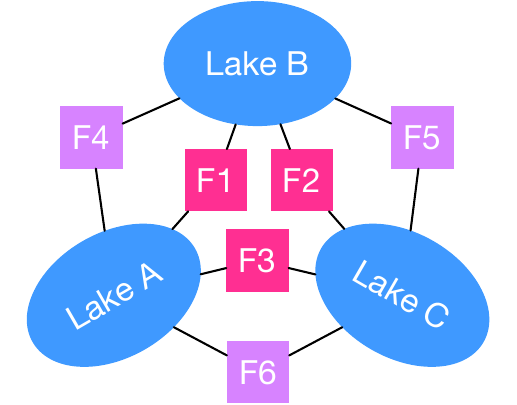}}
        \caption{Two-oversight by two agents}
        \label{intro-lakes-three subfigure}
        \end{center}
    \end{subfigure}
\caption{Six factories example.}
\end{figure}

\begin{proposition}\label{three-agents proposition}
Under the three-agent arrangement depicted in Figure~\ref{intro-lakes-two subfigure}, if the fish in one of the lakes is killed, then \underline{two} agents become responsible for the harmful outcome.    
\end{proposition}
\begin{proof}
We split the decision process into two stages. Stage 1 covers the discharge decisions for the first three factories, F1, F2, and F3. Stage 2 covers the decisions for the remaining factories, F4, F5, and F6. Observe that, at the end of Stage 1, one of the following two cases takes place:

\noindent{\em Case I}: A single portion of the pollutant has been dumped in each of the three lakes. Note that among the remaining factories F4, F5, and F6 only two are connected to each individual lake. Hence, during Stage 2, at most two portions of the pollutant could be added to each lake. Thus, the total amount of the pollutant accumulated in each lake during both stages would not exceed three portions. Therefore, in this case, a harmful outcome is not possible.

\noindent{\em Case II}: Two portions of the pollutant have been dumped in one lake (call it lake X), one portion dumped in another lake (call it lake Y), and no pollution dumped into the third lake (call it lake Z). Observe, similarly to Case I, that at most two portions of the pollutant could be added to each lake during Stage 2. Thus, the fish in lakes Y and Z are guaranteed to survive no matter where factories F4, F5, and F6 dump the pollutant. Hence, to avoid a harmful outcome, it suffices to save the fish in lake X. Out of factories F4, F5, and F6, two are located at the shore of lake X (no matter which of the three lakes is lake X), see Figure~\ref{intro-lakes-two subfigure}. These two factories are controlled by different AI agents (because all three factories F4, F5, and F6 are controlled by different agents). Thus, two different agents can prevent the harm by instructing their respective factories not to dump the pollutant into lake X on Stage 2.
Therefore, these two agents will both be responsible if the harm happens.
\end{proof}
Intuitively, Proposition~\ref{three-agents proposition} means that in the setting of Figure~\ref{intro-lakes-two subfigure}, the three AI agents provide a safety net for one another, overseeing one another's work. We call such an arrangement {\em peer oversight}. 
More generally, we say that a decision mechanism ensures peer $k$-oversight (or just $k$-oversight) if, under that mechanism, at least $k$ agents are responsible for each harmful outcome.
Gap-free mechanisms, such as in Figure~\ref{intro-lake-one subfigure}, ensure one-oversight. The mechanism in Figure~\ref{intro-lakes-two subfigure} ensures two-oversight. The concept of peer $k$-oversight establishes an architectural design principle for trustworthy multiagent AI systems. In addition to protecting against error-prone AI agents, peer two-oversight also provides a defence mechanism against a single rogue agent, while peer $(k+1)$-oversight provides a defence mechanism against a collusion of $k$ rogue agents.  

As the example in Figure~\ref{intro-lakes-two subfigure} illustrates, for a given $k$, sometimes the nodes of a decision mechanism can be redistributed between multiple agents to ensure $k$-oversight. Of course, this is not always possible. For instance, one cannot divide the control over factories in Figure~\ref{intro-lakes-zero subfigure} between some number of agents to ensure two-oversight. In the case when redistribution of control ensuring $k$-oversight exists, one can naturally ask what the {\em minimal} number of agents required to ensure $k$-oversight is. As it turns out, in the case of our example in Figure~\ref{intro-lakes-two subfigure}, two-oversight can be ensured by just two agents. In Figure~\ref{intro-lakes-three subfigure}, we show how this can be done. In this case, agent Pink controls factories F1, F2, and F3, while agent Purple controls factories F4, F5, and F6.

\begin{proposition}\label{two-agents proposition}
Under the two-agent arrangement depicted in Figure~\ref{intro-lakes-three subfigure}, if the fish is killed in one of the lakes, then \underline{two} agents become responsible for the harmful outcome.    
\end{proposition}
\begin{proof}
Agent Pink has an upfront strategy to avoid the harm by arranging the dumps of factories F1, F2, and F3 in such a way that a single portion of the pollutant is dumped by these three factories into each lake. Note that each lake is connected only to two out of three remaining factories (F4, F5, and F6). Thus, no matter what the remaining three factories do, they can dump at most two portions of the pollutant into each lake. Therefore, after all six dumps, each lake will have at most three portions of the pollutant and the fish in all three lakes will survive.

Agent Purple also has an upfront strategy to avoid the harm by arranging the dumps of factories F4, F5, and F6 in such a way that a single portion of the pollutant is dumped by these three factories into each lake. By the argument similar to the one above, this strategy guarantees survival of the fish in all three lakes.

Note that agent Purple, whose factories dump the pollutant after the factories of agent Pink, also has an {\em alternative} strategy to guarantee the survival of the fish in all three lakes. The alternative strategy consists in instructing factories F4, F5, and F6 to do the {\em opposite} of what is done by factories F1, F2, and F3, respectively. For instance, if factory F1 dumps the pollutant into Lake A, then factory F4 should dump its portion of the pollutant into Lake B. Remarkably, this alternative strategy guarantees that only {\em two} portions of the pollutant are dumped into each lake.
\end{proof}
The above proposition shows that, in our six-factory example, two-oversight can be ensured by just two agents. Of course, two distinct agents cannot be responsible in a single-agent setting. Thus, two is the {\em minimal} number of agents required to ensure two-oversight in the six-factory example.

\subsection{Contribution}

In this article, we show that what we observed in the six-factory example is not an exception but a general rule. Specifically, Theorem~\ref{k-oversight theorem} establishes that, in any decision mechanism ensuring $k$-oversight, control over the decision nodes can be redistributed among the agents so that the resulting mechanism uses only $k$ agents while preserving $k$-oversight. Note that if a harmful outcome is reachable, then at least $k$ agents are necessary to achieve $k$-oversight. Hence, the theorem implies that, in such settings, $k$ is precisely the minimum number of agents required whenever $k$-oversight is possible.

In addition to Theorem~\ref{k-oversight theorem}, we present a polynomial-time algorithm that check whether a $k$-oversight mechanism exists. If it exists, the algorithm also constructs a $k$-oversight mechanism using only $k$ agents.

In Section~\ref{related work}, we discuss an alternative concept of responsibility that combines counterfactual responsibility with responsibility for seeing to it and argue that the same results hold for that form of responsibility as well.

It is worth noting that although the decision mechanisms in both our introductory examples can be represented by a tree (see, for example, Figure~\ref{intro-lakes-tree subfigure}), we show our results in a significantly more general setting where the mechanism is a {\em directed graph} that might contain cycles.

\section{Mechanisms, Responsibility, and Peer Oversight}

We start with the formal definition of the game-like model that we use to capture the interaction process between agents.

\begin{definition}\label{transition system}
A (decision) {\bf\em mechanism} is a tuple $M=(V,E,Q,H,A,\{C_a\}_{a\in A})$, where
\begin{enumerate}
    \item $V$ is a finite set of {\bf\em nodes} and $E\subseteq V\times V$ is a set of directed {\bf\em edges}; let $\textit{Sink}_M$ be the set of nodes without outgoing edges; $V\setminus \textit{Sink}_M$ is the set of {\bf\em decision} nodes,
    \item $Q\subseteq V$ is a nonempty set of {\bf\em initial} nodes,
    \item $H\subseteq \textit{Sink}_M$ is a set of {\bf\em harmful} sink nodes, 
    \item $A$ is a set of {\bf\em agents},
    \item $\{C_a\}_{a\in A}$ is a disjoint family of {\bf\em control} sets of decision nodes such that $\bigcup_{a\in A}C_a\subseteq V\setminus \textit{Sink}_M$.
\end{enumerate}
\end{definition}
We omit the subscript in $\textit{Sink}_M$ when its value is clear from the context. An example of a mechanism is depicted in Figure~\ref{intro-lakes-tree subfigure}. In that example, set $V$ consists of the 15 nodes of the tree. The edges in that tree are assumed to be directed from the parent node to the child node. The set $\textit{Sink}$ consists of the 8 leaf nodes of the tree. The set $H$ consists of the nodes $u_1$, $u_2$, $u_3$, $u_4$, $u_5$, and $u_6$. Set $A$ of agents is $\{\text{F1},\text{F2},\text{F3}\}$. Set $C_{\text{F1}}$ contains only the root node $v_1$, set $C_{\text{F2}}$ contains only the nodes $v_2$ and $v_3$, and set $C_\text{F3}$ contains only nodes $v_4$, $v_5$, $v_6$, and $v_7$. Note again that, unlike the mechanism depicted in Figure~\ref{intro-lakes-tree subfigure}, in general, the graph $(V,E)$ might contain cycles.

\begin{definition}
For a given mechanism $(V,E,Q,H,A,\{C_a\}_{a\in A})$,
a {\bf\em strategy} of an agent $a\in A$ is a function $s$ from $C_a$ to $V$ such that $(u,s(u))\in E$ for each decision node $u\in C_a$.       
\end{definition}
An example of a strategy of agent F2 under the mechanism depicted in Figure~\ref{intro-lakes-tree subfigure} is a function that maps node $v_2$ into node $v_4$ and node $v_3$ into node $v_7$. 

By a {\em decision path} we mean any sequence of nodes $u_1,u_2,\dots,u_n$ such that $n\ge 1$ and $(u_i,u_{i+1})\in E$ for each $i<n$. Sequence $v_2,v_4,u_1$ is an example of a decision path for the mechanism depicted in Figure~\ref{intro-lakes-tree subfigure}.

%DEFINE decision path (does not have to start at an initial node)

\begin{definition}\label{avoids harm from a node}
For a given mechanism $(V,E,Q,H,A,\{C_a\}_{a\in A})$,
a {\bf\em strategy} of an agent $a\in A$ {\bf\em avoids harm from a node} $u\in V$ if for each decision path $u_1,\dots,u_n$ from node $u_1=u$ to a harmful sink node $u_n\in H$ there is $i<n$ such that $u_i\in C_a$ and $u_{i+1}\neq s(u_i)$.   
\end{definition}
For the example depicted in Figure~\ref{intro-lakes-tree subfigure}, any strategy $s$ of agent F3 such that $s(v_4)\neq u_1$ (in other words, at node $v_4$, factory F3 dumps the pollutant into Lake C rather than Lake A), avoids harm from node $v_4$. 

The next three lemmas capture important properties of strategies that avoid harm. These lemmas will be used later in the article.

\begin{lemma}\label{no strategy from harm}
For any mechanism,
no agent has a strategy that avoids harm from a harmful sink node.     
\end{lemma}
\begin{proof}
Towards a contradiction, suppose that agent $a$ has a strategy $s$ that avoids harm from a harmful sink node $u\in H$  under a mechanism $(V,E,Q,H,A,\{C_a\}_{a\in A})$. Consider the single-element decision path $u_1$ such that $u_1=u$. Thus, by Definition~\ref{avoids harm from a node}, there exists $i<1$ such that $u_i\in C_a$, which is a contradiction because the decision path contains no elements $u_i$ with $i<1$.
\end{proof}

\begin{lemma}\label{power transition lemma}
For a given mechanism  $(V,E,Q,H,A,\{C_a\}_{a\in A})$,
if a strategy $s$ of an agent $a\in A$ avoids harm from a node $u\notin C_a$, then strategy $s$ also avoids harm from any node $v$ such that $(u,v)\in E$.     
\end{lemma}
\begin{proof}
Consider any decision path $u_1,\dots,u_n$ from node $u_1=v$ to a harmful sink node $u_n\in H$. By Definition~\ref{avoids harm from a node}, it suffices to prove that there is $i<n$ such that $u_i\in C_a$ and $u_{i+1}\neq s(u_i)$.

Towards this proof, consider the decision path $u_0,u_1,\dots,u_n$, where $u_0=u$. By Definition~\ref{avoids harm from a node} and the assumption of the lemma that strategy $s$ avoids harm from node $u$, there must exist $i$, such that $0\le i<n$, $u_i\in C_a$, and $u_{i+1}\neq s(u_i)$. Observe that $i\neq 0$ by the assumption $u\notin C_a$ of the lemma because $u_0=u$. 
\end{proof}

Although the statement of the next lemma differs from the previous one, the proofs of these two lemmas are surprisingly similar.

\begin{lemma}\label{second power transition lemma}
For a given mechanism  $(V,E,Q,H,A,\{C_a\}_{a\in A})$,
if a strategy $s$ of an agent $a\in A$ avoids harm from a node $u\in C_a$, then strategy $s$ also avoids harm from the node $s(u)$.     
\end{lemma}
\begin{proof}
Consider any decision path $u_1,\dots,u_n$ from node $u_1=s(u)$ to a harmful sink node $u_n\in H$. By Definition~\ref{avoids harm from a node}, it suffices to prove that there is $i<n$ such that $u_i\in C_a$ and $u_{i+1}\neq s(u_i)$.

Towards this proof, consider the decision path $u_0,u_1,\dots,u_n$, where $u_0=u$. By Definition~\ref{avoids harm from a node} and the assumption of the lemma that strategy $s$ avoids harm from node $u$, there must exist $i$, such that $0\le i<n$, $u_i\in C_a$, and $u_{i+1}\neq s(u_i)$. Observe that $i\neq 0$ because $u_{0+1}=u_1=s(u)=s(u_0)$.    
\end{proof}

We are now ready to give the core definition of this article, the one that captures the notion of counterfactual responsibility.

\begin{definition}\label{responsible}
An agent is {\bf\em responsible} along a decision path that terminates at a harmful sink node if the agent has a strategy that avoids harm from at least one node along the path.
\end{definition}
For instance, in the setting of Figure~\ref{intro-lakes-tree subfigure}, agent F3 is responsible along the path $v_2,v_4,u_1$ because F3 has a strategy at node $v_4$ to avoid harm.  

\begin{definition}\label{oversight}
For any integer $k\ge 1$, a mechanism ensures (peer) $k${\bf\em -oversight} if there are at least $k$ distinct agents responsible along each decision path from an initial node to a harmful sink node.
\end{definition}
Note that the mechanism ensures 1-oversight iff it has no responsibility gap. Thus, for instance, the mechanism depicted in Figure~\ref{intro-lakes-tree-colours-two subfigure} ensures oversight, while the mechanism depicted in Figure~\ref{intro-lakes-tree-colours subfigure} does not.
Our introduction contains no trees representing the decision mechanism depicted in Figure~\ref{intro-lake-one subfigure}, Figure~\ref{intro-lakes-two subfigure}, and Figure~\ref{intro-lakes-three subfigure}. This is because these trees contain $2^6$ leaf (sink) nodes and, thus, are simply too large to draw. However, as we proved in Proposition~\ref{three-agents proposition} and Proposition~\ref{two-agents proposition}, the last two mechanisms ensure 2-oversight. The single-agent mechanism in Figure~\ref{intro-lake-one subfigure} only ensures 1-oversight.

\begin{definition}\label{structurally equivalent}
Mechanisms $(V,E,Q,H,A,\{C_a\}_{a\in A})$ and
$(V',E',Q',H',A',\{C'_a\}_{a\in A'})$ are {\bf\em structurally equivalent} if $V=V'$, $E=E'$, $Q=Q'$, and $H=H'$.
\end{definition}
The mechanisms depicted in Figure~\ref{intro-lakes-tree-colours subfigure} and Figure~\ref{intro-lakes-tree-colours-two subfigure} are structurally equivalent. And so are those depicted in Figure~\ref{intro-lake-one subfigure}, Figure~\ref{intro-lakes-two subfigure}, and Figure~\ref{intro-lakes-three subfigure}.

\section{First Result: All That You Need Are $k$ Agents}

In this section, we prove our first main result: for any mechanism that ensures $k$-oversight, there is a structurally equivalent mechanism with just $k$ agents that also ensures $k$-oversight. In the rest of the article, by {\em extended} natural numbers we mean the set of non-negative integers extended by the element $\infty$. As is common in mathematics, we assume that: (a) $\infty$ is larger than each of the non-negative integers, (b) $\infty+1=\infty$, and (c) the minimal element of the empty set is $\infty$.

At the core of the proofs of our two main results is the notion of {\em ranking function}, defined below.

\begin{figure}%[ht]
    \centering
       \begin{subfigure}[b]{0.32\textwidth}
        \begin{center}
       \scalebox{0.5}{\includegraphics{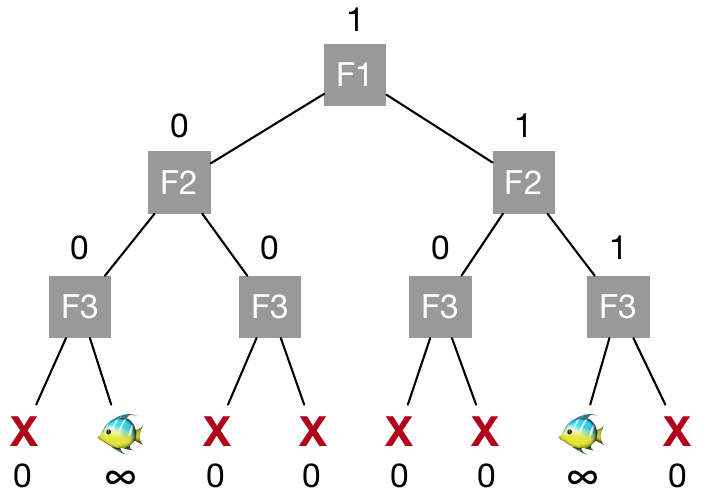}}
        \caption{A ranking function}
        \label{intro-lakes-phi subfigure}
        \end{center}
    \end{subfigure}
\hspace{20mm}
    \begin{subfigure}[b]{0.32\textwidth}
        \begin{center}
        \scalebox{0.5}{\includegraphics{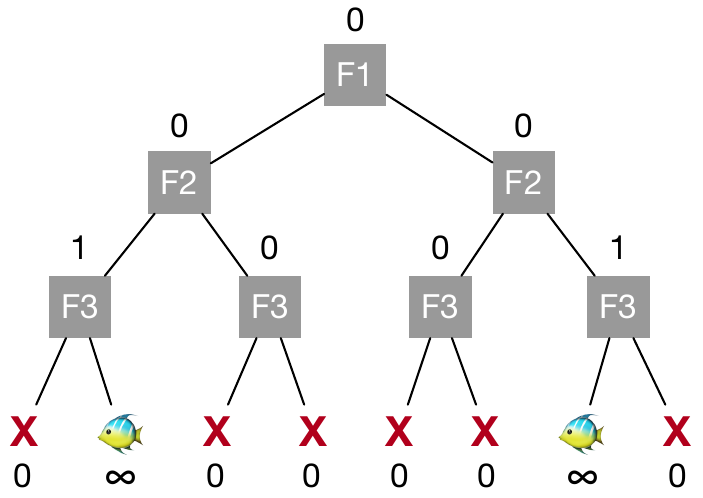}}
        \caption{Responsibility index $\rho(u)$}
       \label{intro-lake-rho subfigure}
        \end{center}
    \end{subfigure}
 
\caption{Two examples of ranking functions for the mechanism in Figure~\ref{intro-lakes-tree subfigure}.}
\end{figure}

\begin{definition}\label{ranking function}
A {\bf\em ranking function} $\phi$ on a mechanism  $(V,E,Q,H,A,\{C_a\}_{a\in A})$ is an arbitrary function that maps each node $u\in V$ into an \underline{extended} natural number $\phi(u)$ such that
\begin{enumerate}
    \item $\phi(u)=0$ for each harmful sink node $u\in H$,
    \item $\phi(u)=\infty$ for each non-harmful sink node $u\in Sink\setminus H$,
    \item for all nodes $u,v\in V$, if $(u,v)\in E$, then $\phi(u)\le \phi(v)+1$,
    \item for each decision node $u\in V\setminus Sink$, there is a node $v\in V$ such that $(u,v)\in E$ and $\phi(u)\le \phi(v)$.
\end{enumerate}
\end{definition}
Figure~\ref{intro-lakes-phi subfigure} and Figure~\ref{intro-lake-rho subfigure} show two examples of ranking functions for the mechanism shown in Figure~\ref{intro-lakes-tree subfigure}. In these figures, the value $\phi(u)$ of a ranking function for a node $u$ is displayed next to the node. 

\begin{lemma}\label{structurally equivalent - ranking function}
Any two structurally equivalent mechanisms have the same set of ranking functions.    
\end{lemma}
\begin{proof}
The statement of the lemma follows from Definition~\ref{structurally equivalent} and Definition~\ref{ranking function} because the latter definition does not refer to the family of control sets $\{C_a\}_{a\in A}$.    
\end{proof}

\begin{definition}\label{count}
For a given mechanism, the {\bf\em responsibility count} of a decision path that terminates at a harmful sink is the number of agents responsible along the path.
\end{definition}
For the mechanism depicted in Figure~\ref{intro-lakes-tree subfigure}, the responsibility count of the decision path $v_2,v_4,u_1$ is 1. This is because the only agent responsible along this path is agent F3 (it has a strategy to avoid harm at node $v_4$). 

\begin{definition}\label{index}
For a given mechanism, the {\bf\em responsibility index} $\rho(u)$ of a node $u$ is the minimum responsibility count among all decision paths from $u$ to a harmful sink node.      
\end{definition}
Of course, there might be no decision paths from node $u$ to a harmful sink node. In this case, $\rho(u)=\infty$ because of our assumption in the preamble to this section that the minimal element of the empty set is $\infty$. Figure~\ref{intro-lake-rho subfigure} shows the index $\rho(u)$ for each node $u$ of the mechanism from Figure~\ref{intro-lakes-tree subfigure}. As we stated earlier, the function $\rho$ in Figure~\ref{intro-lake-rho subfigure} is a ranking function. This is not a coincidence. As we show in Lemma~\ref{index is ranking}, the responsibility index of any decision mechanism is a ranking function. The four lemmas below prove each of the four properties of a ranking function from Definition~\ref{ranking function} separately.

\begin{lemma}
    $\rho(u)=0$ for each harmful sink node $u\in H$ of a mechanism $(V,E,Q,H,A,\{C_a\}_{a\in A})$.
\end{lemma}
\begin{proof}
Consider the single-element decision path $u$. By Lemma~\ref{no strategy from harm}, no agent can have a strategy to avoid harm from any of the nodes along this path. Hence, by Definition~\ref{responsible}, no agent is responsible along this path. Thus, by Definition~\ref{count}, the responsibility count of this path is zero. Therefore, $\rho(u)=0$ by Definition~\ref{index}.
\end{proof}

\begin{lemma}
    $\rho(u)=\infty$ for each non-harmful sink node $u\in \textit{Sink}\setminus H$ of a mechanism $(V,E,Q,H,A,\{C_a\}_{a\in A})$.
\end{lemma}
\begin{proof}
The assumption $u\in \textit{Sink}\setminus H$ implies that there are no paths from $u$ to a harmful sink. Therefore, $\rho(u)=\min\varnothing=\infty$, by Definition~\ref{index}.  
\end{proof}

\begin{lemma}
For any nodes $u,v\in V$ of a decision mechanism $(V,E,Q,H,A,\{C_a\}_{a\in A})$,
if $(u,v)\in E$, then $\rho(u)\le \rho(v)+1$.  
\end{lemma}
\begin{proof} We consider the following two cases separately.

\noindent{\em Case I:} there is no decision path from node $v$ to a harmful sink node. Thus, $\rho(v)=\min\varnothing=\infty$ by Definition~\ref{index}. Therefore,  $\rho(u)\le \infty =\infty+1=\rho(v)+1$. 

\noindent{\em Case II:} there is at least one decision path from node $v$ to a harmful sink node. Consider a decision path $u_1,u_2,\dots,u_n$ from node $u_1=v$ to a harmful sink node $u_n\in H$ with minimal responsibility count. Such a decision path does not have to be unique. Let $k$ be the responsibility count of this decision path. Thus, $\rho(v)=k$ by Definition~\ref{index}. By Lemma~\ref{power transition lemma}, the responsibility count of the path $u,u_1,\dots,u_n$ is at most $k+1$. Thus, $\rho(u)\le k+1$ by Definition~\ref{index}. Therefore, $\rho(u)\le \rho(v)+1$.
\end{proof}

\begin{lemma}
For any non-sink node $u\in V\setminus \textit{Sink}$ of a decision mechanism $(V,E,Q,H,A,\{C_a\}_{a\in A})$, there is node $v\in V$ such that $(u,v)\in E$ and $\rho(u)\le \rho(v)$.     
\end{lemma}
\begin{proof} We start the proof with the following claim:
\begin{claim}
There is a node $v\in V$ such that $(u,v)\in E$ and any agent that has a strategy to avoid harm from node $u$ also has a strategy to avoid harm from $v$.    
\end{claim}
\begin{proof-of-claim}
We consider the following two cases separately:

\noindent{\em Case 1:} There is no agent $a\in A$ such that $u\in C_a$ and agent $a$ has a strategy to avoid harm from node $u$. The assumption $u\in V\setminus \textit{Sink}$ implies that there is at least one node $v\in V$ such that $(u,v)\in E$. Then, the statement of the claim follows from Lemma~\ref{power transition lemma}.   

\noindent{\em Case 2:}  There is at least one agent $a\in A$ such that $u\in C_a$ and agent $a$ has a strategy to avoid harm from node $u$. There could be only one such agent $a$ because the family $\{C_a\}_{a\in A}$ is disjoint by item~5 of Definition~\ref{transition system}. Call this agent $a_0$ and the strategy $s$. Hence, by Lemma~\ref{second power transition lemma}, agent $a_0$ has a strategy to avoid harm from node $v=s(u)$. Furthermore, by Lemma~\ref{power transition lemma}, any agent $a\neq a_0$ that has a strategy to avoid harm from node $u$ also has a strategy to avoid harm from node $v$. 
\end{proof-of-claim}

To finish the proof of the lemma, suppose that $v$ is the node whose existence is proven in the above claim. If there is no decision path from node $v$ to a harmful sink, then $\rho(v)=\infty$ by Definition~\ref{index}. Therefore, $\rho(u)\le \rho(v)$.

Assume now that there is at least one decision path from $v$ to a harmful sink node. Let $u_1,\dots,u_n$, where $u_1=v$ and $u_n\in H$, be such a path with the minimal responsibility count. Let $k$ be the responsibility count of this decision path. Thus, $\rho(v)=k$ by Definition~\ref{index}. By the above claim, Definition~\ref{responsible}, and Definition~\ref{count}, the responsibility count of the path $u,u_1,\dots,u_n$ is also $k$. Therefore, $\rho(u)\le k=\rho(v)$ by Definition~\ref{index}.
This concludes the proof of the lemma.
\end{proof}

The next lemma follows from Definition~\ref{ranking function} and the four previous lemmas.
\begin{lemma}\label{index is ranking}
Responsibility index on an arbitrary mechanism is a ranking function.    
\end{lemma}

The proof of the next lemma contains the main construction of this article. It shows how to build a mechanism with just $k$ agents that ensures $k$-oversight.

\begin{lemma}\label{construction lemma}
For any ranking function $\phi$ on a mechanism $M=(V,E,Q,H,A,\{C_a\}_{a\in A})$ and any integer $k\ge 1$, if $\phi(q)\ge k$ for each initial node $q\in Q$, then there is a mechanism structurally equivalent to $M$ that has only $k$ agents and ensures $k$-oversight.  
\end{lemma}
\begin{proof}
Let $A'=\{1,\dots,k\}$ and, for each agent $a\in A'$,  
\begin{equation}\label{18-aug-a}
    C'_a=\{u\in V\mid \phi(u)=a\}.
\end{equation}
It suffices to show that the mechanism $(V,E,Q,H,A',\{C'_a\}_{a\in A'})$ ensures $k$-oversight. By Definition~\ref{oversight} and Definition~\ref{responsible}, it is enough to show that each agent $a\in A'$ has a strategy to avoid harm from each initial node $q\in Q$. 

Consider an arbitrary agent $a\in A'$. For each node $u\in C'_a$, define $s(u)$ to be any node $v$ such that $(u,v)\in E$ and $\phi(u)\le \phi(v)$. Such a node $v$ exists by item~4 of Definition~\ref{ranking function}. Then, for each $u\in C'_a$,
\begin{equation}\label{18-aug-b}
    \phi(u)\le \phi(s(u)).
\end{equation}

To finish the proof of the lemma, it suffices to establish that strategy $s$ of agent $a$ avoids harm from each initial node. 

Consider any decision path $u_1,\dots,u_n$ from an initial node $u_1\in Q$ to a harmful sink $u_n\in H$. By Definition~\ref{avoids harm from a node}, it is sufficient to prove that there is $i<n$ such that $u_i\in C'_a$ and $u_{i+1}\neq s(u_i)$. Indeed, $\phi(u_1)\ge k$ by the assumption of the lemma. Hence, $\phi(u_1)\ge a$ because $a\in A'=\{1,\dots,k\}$. At the same time, $\phi(u_n)=0$ by item~1 of Definition~\ref{ranking function}. Thus, by item~3 of Definition~\ref{ranking function}, there must exist $i<n$ such that $\phi(u_i)=a$ and $\phi(u_{i+1})=a-1$. Hence, $u_i\in C'_a$ by statement~\eqref{18-aug-a}. Also, $u_{i+1}\neq s(u_i)$ by statement~\eqref{18-aug-b}. 
\end{proof}

The next theorem states the first main result of this article.

\begin{theorem}\label{k-oversight theorem}
For any mechanism that ensures $k$-oversight, there is a structurally equivalent mechanism with just $k$ agents that also ensures $k$-oversight.  
\end{theorem}
\begin{proof}
If a mechanism ensures $k$-oversight, then $\rho(u)\ge k$ for each initial node $u$ of the mechanism by Definition~\ref{oversight}, Definition~\ref{responsible}, and Definition~\ref{index}.
Then, the statement of the theorem follows from Lemma~\ref{index is ranking} and Lemma~\ref{construction lemma}.   
\end{proof}

\section{Second Result: Mechanism Ensuring $k$-oversight is Feasibly Computable}

In this section, we present our second result: a polynomial-time algorithm that, for any given mechanism, either constructs a structurally equivalent mechanism ensuring $k$-oversight or determines that such a mechanism does not exist. To construct such an algorithm, let us first introduce the notion of a closure $X^+$ of a set of nodes $X$. Intuitively, $X^+$ is the set of nodes from which the mechanism will inevitably reach a node in the set $X$. For the mechanism depicted in Figure~\ref{intro-lakes-tree subfigure}, if $X=\{v_4,u_2,u_3,v_3\}$, then $X^+=X\cup \{v_1,v_2,v_5\}$. Formally, the set $X^+$ is defined recursively:

\begin{definition}\label{plus}
For any set $X\subseteq V$ of nodes of a mechanism $M=(V,E,Q,H,A,\{C_a\}_{a\in A})$, 
let $X^+_M$ be the minimal set of nodes such that
\begin{enumerate}
    \item $X\subseteq X^+_M$,
    \item if $u\notin \textit{Sink}$ and $\{v\in V\mid (u,v)\in E\}\subseteq X^+_M$, then $u\in X^+_M$.
\end{enumerate}
\end{definition}
We omit the subscript in the expression $X^+_M$ when its value is clear from the context. The next two lemmas follow from the above definition.

\begin{lemma}\label{plus no sinks}
For any mechanism $(V,E,Q,H,A,\{C_a\}_{a\in A})$, sink node $u\in Sink$, and set $X\subseteq V$ of nodes, if $u\in X^+$, then $u\in X$.     
\end{lemma}

\begin{lemma}\label{non-plus child}
For any mechanism $(V,E,Q,H,A,\{C_a\}_{a\in A})$, decision node $u\in V\setminus Sink$, and set $X\subseteq V$ of nodes such that $u\notin X^+$, there is a node $v\in V$ such that $(u,v)\in E$ and  $v\notin X^+$.  
\end{lemma}

The next lemma captures a connection between the closure operation $X^+$ and ranking functions.
\begin{lemma}\label{phi plus}
For any mechanism $(V,E,Q,H,A,\{C_a\}_{a\in A})$, ranking function $\phi$, set $X\subseteq V$ of nodes, and node $u\in X^+$,
$$
\phi(u)\le \max_{v\in X}\phi(v).
$$
\end{lemma}
\begin{proof}
We prove the statement of the lemma by induction on the recursive construction of the set $X^+$ specified in Definition~\ref{plus}.

\vspace{1mm}\noindent{\em Base Case}: $u\in X$. Then, $\phi(u)\le\max_{v\in X}\phi(v)$.

\vspace{1mm}\noindent{\em Induction Step}: $u\notin \textit{Sink}$ and $\{v\in V\mid (u,v)\in E\}\subseteq X^+$. Then, $\phi(v)\le\max_{w\in X}\phi(w)$ for each $v\in V$ such that $(u,v)\in E$ by the induction hypothesis. At the same time, by item~4 of Definition~\ref{ranking function}, the assumption $u\notin \textit{Sink}$ implies that there is a node $v_0\in V$ such that $(u,v_0)\in E$ and $\phi(u)\le \phi(v_0)$. Therefore,
$\phi(u)\le \phi(v_0)\le\max_{w\in X}\phi(w)$.
\end{proof}

\begin{definition}\label{tl}
For any mechanism $M=(V,E,Q,H,A,\{C_a\}_{a\in A})$ and set $X\subseteq V$ of nodes, let $tl_M(X)$ be the set of all nodes $u\in V$ for which there is a node $v\in X$ such that $(u,v)\in E$.  
\end{definition}
We omit the subscript in the expression $tl_M(X)$ when its value is clear from the context. In the setting of Figure~\ref{intro-lakes-tree subfigure}, we have $tl(\{u_1,v_4\})=\{v_4,v_2\}$.
The next lemma follows from the above definition.

\begin{lemma}\label{tl no sinks}
For any mechanism $(V,E,Q,H,A,\{C_a\}_{a\in A})$, sink node $u\in Sink$, and set $X\subseteq V$ of nodes,
$u\notin tl(X)$.
\end{lemma}

\begin{definition}\label{X}
For any mechanism $M=(V,E,Q,H,A,\{C_a\}_{a\in A})$, let sequence 
$\mathbb{X}_0^M,\mathbb{X}_1^M,\dots,\mathbb{X}_\infty^M$ of sets of nodes be defined recursively as follows:
$$
\mathbb{X}_i^M=
\begin{cases}
    H^+, & \text{if $i=0$},\\
    (\mathbb{X}_{i-1}^M\cup tl(\mathbb{X}_{i-1}^M))^+,& \text{if $0<i<\infty$},\\
    V, & \text{if $i=\infty$}.\\
\end{cases}
$$
\end{definition}
We omit the superscript in the expression $\mathbb{X}_i^M$ when its value is clear from the context. For the example depicted in Figure~\ref{intro-lakes-tree subfigure},
\begin{align*}
   \mathbb{X}_0=&H^+=\{u_1,u_2,u_3,u_4,u_5,u_6,v_5,v_6\},\\
   \mathbb{X}_1=&(\mathbb{X}_0\cup tl(\mathbb{X}_0))^+=(\{u_1,u_2,u_3,u_4,u_5,u_6,v_5,v_6\}\cup\{v_2,v_3,v_4,v_5,v_6,v_7\})^+\\
   =&\{u_1,u_2,u_3,u_4,u_5,u_6,v_1,v_2,v_3,v_4,v_5,v_6,v_7\},\\
   \mathbb{X}_2=&(\mathbb{X}_1\cup tl(\mathbb{X}_1))^+=\mathbb{X}_1,\\
   \mathbb{X}_3=&(\mathbb{X}_2\cup tl(\mathbb{X}_2))^+=\mathbb{X}_1,\\
   \dots&
\end{align*}
and $\mathbb{X}_\infty$ is the set of all nodes of the tree, including the two unlabelled leaf nodes.
\begin{lemma}
$\mathbb{X}_0\subseteq \mathbb{X}_1\subseteq \mathbb{X}_2\subseteq\dots \subseteq \mathbb{X}_\infty$
for any mechanism $(V,E,Q,H,A,\{C_a\}_{a\in A})$.
\end{lemma}
\begin{proof}
By Definition~\ref{plus}, if $X\subseteq Y$, then $X^+\subseteq Y^+$ for all sets $X,Y\subseteq V$.    
By Definition~\ref{tl}, if $X\subseteq Y$, then $tl(X)\subseteq tl(Y)$ for all sets $X,Y\subseteq V$.
Therefore, the statement of the lemma follows from Definition~\ref{X}.
\end{proof}

The next lemma connects the chain of sets $\mathbb{X}_0\subseteq \mathbb{X}_1\subseteq\dots $ with an {\em arbitrary} ranking function.
\begin{lemma}\label{phi le k}
$\phi(u)\le k$,
for any mechanism $(V,E,Q,H,A,\{C_a\}_{a\in A})$, ranking function $\phi$,    
extended natural number $k$ such that $0\le k\le \infty$, and node $u\in \mathbb{X}_k$. 
\end{lemma}
\begin{proof}
We prove the statement of the lemma by induction on $k$.

\vspace{1mm}
\noindent{\em Base Case}: $k=0$. Then, $u\in H^+$ by the assumption $u\in \mathbb{X}_k$ and Definition~\ref{X}. Hence,
$\phi(u)\le \max_{v\in H}\phi(v)$ by Lemma~\ref{phi plus}. 
Thus, $\phi(u)\le 0$ by item~1 of Definition~\ref{ranking function}.

\vspace{1mm}
\noindent{\em Induction Step}: $0<k<\infty$. Then,
$u\in (\mathbb{X}_{k-1}\cup tl(\mathbb{X}_{k-1}))^+$ by the assumption $u\in \mathbb{X}_k$ and Definition~\ref{X}. Hence, by Lemma~\ref{phi plus},
$$
\phi(u)\le \max\{\phi(v)\mid v\in \mathbb{X}_{k-1}\cup tl(\mathbb{X}_{k-1})\}  
=\max\left( \max\{\phi(v)\mid v\in \mathbb{X}_{k-1}\}, 
\max\{\phi(v)\mid v\in tl(\mathbb{X}_{k-1})\}
\right).
$$
Thus, 
$
\phi(u)\le \max\left( k-1, 
\max\{\phi(v)\mid v\in tl(\mathbb{X}_{k-1})\}
\right)
$
by the induction hypothesis.
Then,
by Definition~\ref{tl},
$$
\phi(u)\le \max\left( k-1, 
\max\{\phi(v)\mid (v,w)\in E, w\in \mathbb{X}_{k-1}\}
\right).
$$
Hence, 
$
\phi(u)\le \max\left( k-1, 
\max\{\phi(v)\mid (v,w)\in E, \phi(w)\le k-1\}
\right)
$
again by the induction hypothesis.
Thus, 
$
\phi(u)\le \max\left( k-1, (k-1)+1
\right)
$
by item~3 of Definition~\ref{ranking function}.
Therefore, $\phi(u)\le k$.

\vspace{1mm}
\noindent{\em Infinity Case}: $k=\infty$. Note that $\phi(u)\le \infty$. Thus, $\phi(u)\le k$ because $k=\infty$.
\end{proof}

\begin{definition}\label{iota}
For any mechanism $M=(V,E,Q,H,A,\{C_a\}_{a\in A})$, the {\bf\em irrecoverability index} $\iota_M(u)$ of a node $u\in V$ is the minimal $i$ such that $u\in \mathbb{X}_i$.   
\end{definition}
Note that the irrecoverability index is well-defined because $\mathbb{X}_\infty=V$ by Definition~\ref{X}.
We omit the subscript in the expression $\iota_M(u)$ when its value is clear from the context. 
   \begin{figure}
        \begin{center}
       \scalebox{0.5}{\includegraphics{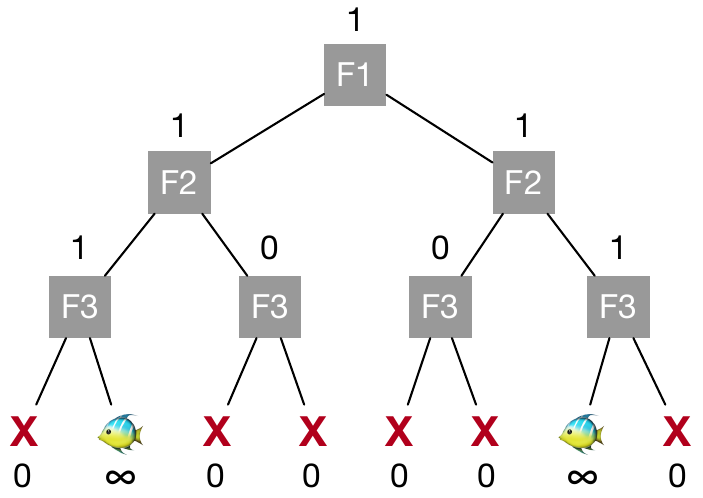}}
        \caption{Irrecoverability index $\iota(u)$.}
        \label{intro-lakes-iota subfigure}
        \end{center}
    \end{figure}
Figure~\ref{intro-lakes-iota subfigure} shows the irrecoverability index for each node of the mechanism depicted in Figure~\ref{intro-lakes-tree subfigure}. It is easy to see that, in this example, the irrecoverability index is a ranking function. This is also not a coincidence. As we prove in Lemma~\ref{iota is ranking}, the irrecoverability index of any mechanism is a ranking function. The four lemmas below prove each of the four properties of a ranking function from Definition~\ref{ranking function} separately.

\begin{lemma}
    $\iota(u)=0$ for each harmful sink node $u\in H$ of a mechanism $(V,E,Q,H,A,\{C_a\}_{a\in A})$.
\end{lemma}
\begin{proof}
Note that $H\subseteq H^+$ by Definition~\ref{plus}. Hence, $H\subseteq\mathbb{X}_0$ by Definition~\ref{X}. Therefore, $\iota(u)=0$ for each node $u\in H$ by Definition~\ref{iota}.   
\end{proof}

\begin{lemma}
    $\iota(u)=\infty$ for each non-harmful sink node $u\in \textit{Sink}\setminus H$ of a mechanism $(V,E,Q,H,A,\{C_a\}_{a\in A})$.
\end{lemma}
\begin{proof}
Consider any node $u\in \textit{Sink}\setminus H$. By Definition~\ref{iota}, it suffices to show that $u\notin \mathbb{X}_i$ for each $i<\infty$. We prove this by induction on $i$.
For the base case, note that $u\notin H^+$ by Lemma~\ref{plus no sinks} and the assumption $u\in \textit{Sink}\setminus H$. Hence, $u\notin \mathbb{X}_0$ by Definition~\ref{X}.

Consider any $i>0$ and suppose that $u\notin \mathbb{X}_i$. Thus, $u\notin \mathbb{X}_i\cup tl(\mathbb{X}_i)$ by Lemma~\ref{tl no sinks} and the assumption $u\in \textit{Sink}\setminus H$. Thus, $u\notin (\mathbb{X}_i\cup tl(\mathbb{X}_i))^+$ by Lemma~\ref{plus no sinks}. Therefore, $u\notin \mathbb{X}_{i+1}$ by Definition~\ref{X}.
\end{proof}

\begin{lemma}
If $(u,v)\in E$, then $\iota(u)\le \iota(v)+1$, for any $u,v\in V$ of a mechanism $(V,E,Q,H,A,\{C_a\}_{a\in A})$.  
\end{lemma}
\begin{proof}
Note that
$v\in \mathbb{X}_{\iota(v)}$ 
by Definition~\ref{iota}.
Thus, $u\in tl(\mathbb{X}_{\iota(v)})$
by the assumption $(u,v)\in E$ and Definition~\ref{tl}.
Hence $u\in \mathbb{X}_{\iota(v)}\cup tl(\mathbb{X}_{\iota(v)})$.
Then, $u\in (\mathbb{X}_{\iota(v)}\cup tl(\mathbb{X}_{\iota(v)}))^+$ by item~1 of Definition~\ref{plus}.
Thus,
$u\in \mathbb{X}_{\iota(v)+1}$ by Definition~\ref{X}.
Therefore, $\iota(u)\le \iota(v)+1$ by Definition~\ref{iota}.
\end{proof}

\begin{lemma}
For each decision node $u\in V\setminus Sink$ of a mechanism $(V,E,Q,H,A,\{C_a\}_{a\in A})$, there is node $v\in V$ such that $(u,v)\in E$ and $\iota(u)\le \iota(v)$.     
\end{lemma}
\begin{proof} We consider the following four cases separately:

\vspace{1mm}
\noindent{\em Case 1:} $\iota(u)=0$. The assumption $u\in V\setminus Sink$ implies that there is $v\in V$ such that $(u,v)\in E$. Then, $\iota(u)=0\le \iota(v)$.

\vspace{1mm}
\noindent{\em Case 2:}  $\iota(u)=1$. Thus, $u\notin\mathbb{X}_{0}$ by Definition~\ref{iota}.
Hence, $u\notin H^+$ by Definition~\ref{X}. Then, by Lemma~\ref{non-plus child} and the assumption $u\in V\setminus Sink$ of the lemma, there exists a node $v\in V$ such that $(u,v)\in E$ and $v\notin H^+$.
Hence, $v\notin \mathbb{X}_0$ by Definition~\ref{X}. Thus, $\iota(v)>0$ by Definition~\ref{iota}. Therefore, $\iota(v)\ge\iota(u)$ by the assumption $\iota(u)=1$ of the case.

\vspace{1mm}
\noindent{\em Case 3:}  $1< \iota(u)<\infty$.
Then, $u\notin\mathbb{X}_{\iota(u)-1}$ by Definition~\ref{iota}.
Hence, $u\notin (\mathbb{X}_{\iota(u)-2}\cup tl(\mathbb{X}_{\iota(u)-2}))^+$ by Definition~\ref{X} and the assumption $1< \iota(u)$ of the case. Thus, by Lemma~\ref{non-plus child} and the assumption $u\in V\setminus Sink$ of the current lemma, there exists a node $v\in V$ such that $(u,v)\in E$ and $v\notin (\mathbb{X}_{\iota(u)-2}\cup tl(\mathbb{X}_{\iota(u)-2}))^+$.
Then,
$v\notin\mathbb{X}_{\iota(u)-1}$
by Definition~\ref{X}.
Hence,
$\iota(v)>\iota(u)-1$
by Definition~\ref{iota}.
Therefore,
$\iota(v)\ge \iota(u)$.

\vspace{1mm}
\noindent{\em Case 4:}  $\iota(u)=\infty$.
The assumption
$u\in V\setminus Sink$
of the lemma implies that there is a node $v\in V$ such that $(u,v)\in E$. It suffices to show that $\iota(v)=\infty$. Towards a contradiction, suppose that $\iota(v)<\infty$. 
Note that $u\in tl(\mathbb{X}_{\iota(v)})$ by the assumption $(u,v)\in E$ and Definition~\ref{tl}.
Thus, $u\in \mathbb{X}_{\iota(v)}\cup tl(\mathbb{X}_{\iota(v)})$. Hence,
$u\in (\mathbb{X}_{\iota(v)}\cup tl(\mathbb{X}_{\iota(v)}))^+$ by item~1 of Definition~\ref{plus}.
Then, $u\in \mathbb{X}_{\iota(v)+1}$ by Definition~\ref{X}.
Therefore, $\iota(u)\le\iota(v)+1<\infty$ by the assumption $\iota(v)<\infty$, which contradicts the assumption  $\iota(u)=\infty$ of the case. 
\end{proof}

The next lemma follows from Definition~\ref{ranking function} and the four previous lemmas.
\begin{lemma}\label{iota is ranking}
For any mechanism, the irrecoverability index is a ranking function.
\end{lemma}

Recall that Figure~\ref{intro-lakes-phi subfigure}, Figure~\ref{intro-lake-rho subfigure}, and Figure~\ref{intro-lakes-iota subfigure} visualise three different ranking functions for the mechanism depicted in Figure~\ref{intro-lakes-tree subfigure}. It is easy to see that, for each node of the tree, the value of the function in Figure~\ref{intro-lakes-iota subfigure} is at least as large as the corresponding values of the two other ranking functions. This too is not a coincidence. The next lemma proves that all ranking functions on all mechanisms are bounded by the irrecoverability index. Since, by Lemma~\ref{iota is ranking}, the irrecoverability index itself is a ranking function, this means that the irrecoverability index is equal to the pointwise maximum of all ranking functions. 

\begin{lemma}\label{phi le iota}
For each ranking function $\phi$ and node $u$ of a mechanism,
$\phi(u)\le \iota(u)$.
\end{lemma}
\begin{proof}
The statement of the lemma follows from Lemma~\ref{phi le k} and Definition~\ref{iota}.
\end{proof}

We are finally ready to state and prove our second main result: the existence of a computationally efficient procedure for finding a $k$-agent mechanism providing $k$-oversight. 

\begin{theorem}\label{second theorem}
For any mechanism $M$ and integer $k\ge 1$, it can be determined in polynomial time whether there exists a mechanism $M'$ structurally equivalent to $M$ that ensures $k$-oversight. If such a mechanism $M'$ exists, it can be constructed in polynomial time as well. 
\end{theorem}
\begin{proof}
The family of sets $\{\mathbb{X}_i\}_i$ for the mechanism $M=(V,E,Q,H,A,\{C_a\}_{a\in A})$ can be constructed in polynomial time using the recursive construction from Definition~\ref{X}. Thus, by Definition~\ref{iota}, the value $\iota(u)$ for each node $u\in V$ can be computed in polynomial time. We consider the following two cases separately:

\vspace{1mm}\noindent
{\em Case 1}: $\min_{q\in Q}\iota(q)<k$. In this case, mechanism $M'$ does not exist. To prove this, assume that a mechanism $M'=(V,E,Q,H,A',\{C'_a\}_{a\in A'})$ ensures $k$-oversight and is structurally equivalent to mechanism $M$. Let $\rho'$ be the responsibility index of mechanism $M'$. By Lemma~\ref{index is ranking}, function $\rho'$ is a ranking function for mechanism $M'$. By Lemma~\ref{structurally equivalent - ranking function}, function $\rho'$ is also a ranking function for mechanism $M$. Thus, $\rho'(u)\le \iota(u)$ for each node $u\in V$  by Lemma~\ref{phi le iota}. Hence, $\min_{q\in Q}\rho'(q)<k$ by the assumption $\min_{q\in Q}\iota(q)<k$ of the case. Then, there is an initial node $q_0\in Q$ such that $\rho'(q_0)<k$. Thus, by Definition~\ref{index}, there is a decision path $u_1,\dots,u_n$ from the initial node $u_1=q_0$ to a harmful sink node $u_n\in H$ whose responsibility count is less than $k$ under mechanism $M'$. Hence, by Definition~\ref{count}, fewer than $k$ agents are responsible along the path $u_1,\dots,u_n$ under mechanism $M'$. Therefore, by Definition~\ref{oversight}, mechanism $M'$ does not ensure $k$-oversight.

\vspace{1mm}\noindent
{\em Case 2}: $\min_{q\in Q}\iota(q)\ge k$. Note that $\iota$ is a ranking function by Lemma~\ref{iota is ranking}. Thus, by Lemma~\ref{construction lemma}, there is a mechanism structurally equivalent to $M$ that has only $k$ agents and ensures $k$-oversight. The construction of this mechanism, specified by equation~\eqref{18-aug-a}, can be accomplished in polynomial time because the value $\iota(u)$ for each node $u\in V$ can be computed in polynomial time.
\end{proof}

Note that the construction in the proof of Theorem~\ref{second theorem} produces a structurally equivalent mechanism under which some of the nodes might not be controlled by any agents (that is, $\bigcup_{a\in A}C_a\subsetneq V\setminus \textit{Sink}$). Intuitively, such nodes can be viewed as controlled by nature or the environment. However, if desired, these nodes can be put under the control of any of the agents. This modification preserves $k$-oversight. 

\section{Related Work}\label{related work}

Counterfactual responsibility is probably the most widely discussed responsibility concept~\citep{w17}. However, at least four other formal definitions of responsibility have been considered. One of them is the notion of {\em responsibility for seeing to harm}~\citep{b11jal,b11jpl,nt21ijcai,nt23apal,s24aaai}, extensively studied in STIT logic~\citep{bp90krdr,h01,h95jpl,hp17rsl,ow16sl} under the name ``deliberatively seeing to it''. In the setting where agents act sequentially, as in our Definition~\ref{transition system}, an agent is responsible for seeing to harm if the agent took an action that eliminated the last possibility of a non-harmful outcome. For example, although no agent is responsible counterfactually along the decision path $v_1,v_2,v_5,u_2$ in Figure~\ref{intro-lakes-tree subfigure}, factory F2 is responsible for seeing to harm along that decision path. This is because, along that path, the action ``dump into lake C'' by factory F2 eliminated the last possibility that the fish in all three lakes survive. In sequential mechanisms, at most one agent could be the one who eliminates the last possibility of a non-harmful outcome. Thus, in such mechanisms, at most one agent can be responsible along any decision path. As a result, for $k\ge 2$, it is impossible to construct a sequential decision mechanism that ensures $k$-oversight. 

One can also consider a ``bimodal'' responsibility: an agent is bimodally responsible along a decision path if the agent is responsible either counterfactually or for seeing to harm. For any decision path, {\em bimodal} responsibility count can be one more than counterfactual responsibility count (see Definition~\ref{count}) because, along the path, in addition to the agents responsible counterfactually, there could be an additional agent responsible for seeing to harm. Nevertheless, it is relatively easy to see that Lemma~\ref{index is ranking} holds for {\em bimodal} responsibility index. Thus, by Lemma~\ref{phi le iota}, bimodal responsibility index is no more than the irrecoverability index. Therefore, bimodal responsibility cannot ensure a higher degree of oversight than we produce in Theorem~\ref{second theorem} using counterfactual responsibility alone. In the case of our running example from Figure~\ref{intro-lakes-tree subfigure}, the bimodal responsibility index is simply equal to the irrecoverability index. Both are shown in Figure~\ref{intro-lakes-iota subfigure}.

Another commonly discussed responsibility concept is the best-effort responsibility~\citep{bh18ej,n26aaai}. As an example, imagine that three agents concurrently (say by a paper ballot) vote on whether to spare the life of a cat. The decision is made by the majority vote. Suppose that all three agents voted to kill the cat. The cat is dead. Who is responsible for this? None of the agents had a strategy to save the life of the cat. Thus, none of the agents is responsible counterfactually. At the same time, no individual vote guarantees the death of the cat. Hence, none of the agents is responsible for seeing to the death of the cat. Nevertheless, most people will agree that all three agents are intuitively responsible for the death. This intuition can be captured by the best-effort responsibility. Each agent is responsible because the cat is dead and did not use the action (vote to spare the cat's life) that constitutes the best effort to save the cat. Best-effort actions can be defined as non-weakly dominated actions, assuming that a harmful outcome has utility -1 and a non-harmful outcome has utility 0. 
Note that best-effort responsibility is only meaningful in the setting where agents take actions just once and all of them act simultaneously. It is not clear what a best-effort action could be in the settings of Definition~\ref{transition system}. The situation here is similar to classical game theory where the notion of weakly dominated strategy exists for strategic games and does not exist for extensive form games. Best-effort responsibility could also be considered in a probabilistic setting, where an agent might not be minimising the expected chances of a harmful outcome~\citep{dp22scw}. We do not discuss best-effort responsibility in this article because it is not compatible with the class of decision mechanisms that we consider.

Finally, \cite{s24aaai} introduced the notion of a higher-order responsibility. This notion has also been discussed in~\citep{jn26aaai}. An agent $a$ is second-order responsible for the harmful outcome if no agent is counterfactually responsible for the harm and the agent $a$ had a strategy that would guarantee that if harm happens then at least one agent is counterfactually responsible for it. One can similarly introduce third-, fourth-, and higher-order responsibility. We leave the study of $k$-oversight under higher-order responsibility as a direction for future research.

In addition to the works on formally defining various responsibility concepts, there are three papers~\citep{bh18ej,dp22scw,nt25ijcai} that show that gap-free mechanisms do not exist in some special classes of mechanisms. None of them studies how to construct a mechanism with desirable responsibility-related properties, as we do in Theorem~\ref{second theorem}. 

Finally, the concept of $k$-oversight is closely related to the concept of separation of duty in information security. The latter is a design principle stating that certain crucial steps of a process are carried out by different users~\citep*{gfp25acm-cs}.

\section{Conclusion}

This article introduced peer $k$-oversight as a responsibility-related property of sequential collective decision mechanisms. Rather than asking only who is responsible after a harmful outcome occurs, peer oversight treats responsibility as a design requirement: control over the decision process should be distributed so that at least $k$ agents are responsible whenever harm occurs.

The main results show that whenever $k$-oversight can be achieved by redistributing control over the decision nodes, it can be achieved using exactly $k$ agents. Moreover, whether such a redistribution exists can be determined in polynomial time, and a corresponding $k$-agent mechanism can be constructed efficiently.

These results suggest peer oversight as an architectural design principle for collective and multiagent decision-making systems, including multiagent AI systems in which agents are intended to oversee one another. A natural direction for future work is to allow the required degree of peer oversight to depend on the harmful outcome. Since harmful outcomes can differ in severity, more serious harms may warrant stronger oversight, with the required value of $k$ varying across harmful sink nodes.

\bibliographystyle{elsarticle-harv}
\bibliography{naumov}

\end{document}